\documentclass[a4paper,english,numberwithinsect,nosubfigcap,noalgorithm2e]{lipics-v2021}

\usepackage{microtype} 

\usepackage{thm-restate}
\usepackage[basic]{complexity}
\usepackage[capitalise,noabbrev]{cleveref}
\usepackage{mathtools}
\usepackage[ruled,linesnumbered,vlined]{algorithm2e}

\usepackage{todonotes}

\newcommand{\crossings}{\ensuremath{\textrm{cr}}}
\newcommand{\gaps}{\ensuremath{\textrm{gaps}}}
\newcommand{\dummy}{\ensuremath{V_2^{\text{dm}}}}
\newcommand{\real}{\ensuremath{V_2^{\text{r}}}}
\newcommand{\oscmsg}{\textsc{OSCM-SG}}
\newcommand{\oscmkg}{\textsc{OSCM-$k$G}}
\newcommand{\oscm}{\textsc{OSCM}}

\newcommand{\realSuperscript}{\text{r}}
\newcommand{\dummySuperscript}{\text{dm}}

\newcommand{\vTwoReal}{V_2^\realSuperscript}

\newcommand{\permComesBefore}[3]{#2\prec_{#1}#3}

\newcommand{\ilpEquationSpacing}{\quad}

\newcommand{\floor}[1]{\lfloor #1 \rfloor}

\theoremstyle{definition}
\newtheorem{problem}{Problem}
\newtheorem{ilp}{Integer Linear Program}
\theoremstyle{plain}

\crefname{lem}{Lemma}{Lemmas}
\crefname{algorithm}{Algorithm}{Algorithms}
\crefname{algocf}{Algorithm}{Algorithms}

\nolinenumbers

\title{Layered Graph Drawing with Few Gaps and Few Crossings}
\titlerunning{Layered Graph Drawing with Few Gaps and Few Crossings}

\author{Alexander Dobler}{TU Wien, Austria}{adobler@ac.tuwien.ac.at}{0000-0002-0712-9726}{Supported by the Vienna Science and Technology Fund (WWTF) under grant 10.47379/ICT19035}
\author{Jakob Roithinger}{TU Wien, Austria}{jakob.roithinger@student.tuwien.ac.at}{}{}

\authorrunning{A. Dobler and J. Roithinger}

\hideLIPIcs



\begin{document}

\maketitle

\begin{abstract}
  We consider the task of drawing a graph on multiple horizontal layers, where each node is assigned a layer, and each edge connects nodes of different layers. Known algorithms determine the orders of nodes on each layer to minimize crossings between edges, increasing readability. Usually, this is done by repeated one-sided crossing minimization for each layer.
These algorithms allow edges that connect nodes on non-neighboring layers, called ``long'' edges, to weave freely throughout layers of the graph, creating many ``gaps'' in each layer. As shown in a recent work on hive plots -- a similar visualization drawing vertices on multiple layers -- it can be beneficial to restrict the number of such gaps. We extend existing heuristics and exact algorithms for one-sided crossing minimization in a way that restricts the number of allowed gaps.
The extended heuristics maintain approximation ratios, and in an experimental evaluation we show that they perform well with respect to the number of resulting crossings when compared with exact ILP formulations.
\end{abstract}

\section{Introduction}\label{sec:intro}
Drawing graphs is a non-trivial task, and many visualization approaches exist. One such approach, known as \emph{layered graph drawing}, draws the nodes on horizontal layers $L=\{L_1,L_2,\dots,L_\ell\}$, each edge connects nodes of different layers.
Sugiyama et al.\ pioneered the automation of such drawings \cite{sugiyama1981} in the well-known \emph{Sugiyama-framework} consisting of multiple steps. 
The first step assigns nodes to the $\ell$ layers such that nodes connected by an edge are on different layers (\cref{fig:layeredgraph}). In the next step so-called \emph{long edges} connecting nodes $u$ and $v$ of non-neighbouring layers $L_i$ and $L_j$, $i<j-1$, are replaced by a path of length $j-i$. The newly created \emph{dummy} nodes are assigned to layers $i+1,i+2,\dots,j-1$ (\cref{fig:dummynodes}). Original nodes are called \emph{real nodes}. After this process each edge connects nodes of adjacent layers.
In the next step edge crossings are reduced by permuting the nodes of each layer.
Usually, this is performed for neighboring layers, whereby the order of nodes in one layer is fixed and the other layer is permuted.
This is known as one-sided crossing minimization (\oscm) \cite{eades1994oscmNPcomplete}.
\oscm\ is performed iteratively, ``up'' and ``down'' the layers of the graph, i.e.\ for $i=1,2,\dots,\ell-1$ layer $L_i$ is fixed and layer $L_{i+1}$ is permuted.
Then, for $i=\ell,\ell-1,\dots,2$, layer $L_i$ is fixed and layer $L_{i-1}$ is permuted.
Several such ``up'' and ``down'' runs may be performed until reaching a termination condition.
The last step replaces dummy nodes by the original edges, and assigns $x$-coordinates to nodes. 

We are concerned with the second step of the above framework. In existing algorithms, dummy and real nodes are treated equally during the crossing minimization step. This can lead to many gaps in the resulting visualization in each layer. Formally, a gap in layer $L_i$ is a maximal consecutive sequence of dummy nodes (\cref{fig:gaps}). We argue that this hinders readability; thus, we extend algorithms for \oscm\ to only allow (1) \emph{side gaps}, that is, one gap on the left and one gap on the right of a layer, or (2) at most \emph{$k$ gaps} for each layer. 
Gaps have already been motivated by Nöllenburg and Wallinger for hive plots~\cite{DBLP:journals/jgaa/NollenburgW24}, which is essentially a circular variant of the Sugiyama framework with some additional features. Nöllenburg and Wallinger have introduced gaps at fixed positions of each layer, including the variant of side gaps. We extend their work by introducing the variant of $k$ gaps, where the gaps can be placed arbitrarily. Furthermore, we consider the problem from a more theoretical perspective, proving approximation ratios of our algorithms.
For both variants, side gaps and $k$ gaps, we propose approximations and exact algorithms that are experimentally evaluated.
\begin{figure}[!t]
     \centering
     \begin{subfigure}[b]{0.28\textwidth}
         \centering
         \includegraphics[width=\textwidth]{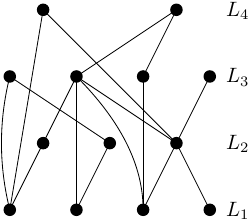}
         \caption{}
         \label{fig:layeredgraph}
     \end{subfigure}
     \hfill
     \begin{subfigure}[b]{0.28\textwidth}
         \centering
         \includegraphics[width=\textwidth]{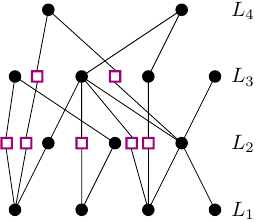}
         \caption{}
         \label{fig:dummynodes}
     \end{subfigure}
     \hfill
     \begin{subfigure}[b]{0.34\textwidth}
         \centering
         \includegraphics[width=\textwidth]{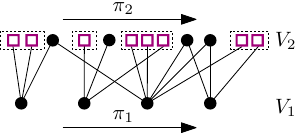}
         \caption{}
         \label{fig:gaps}
     \end{subfigure}
        \caption{(a) A layered graph drawing. (b) Long edges are replaced by paths of dummy nodes, shown as violet squares. (c) A drawing of two layers with the two node orderings $\pi_1$ and $\pi_2$ such that $\pi_2$ has $4$ gaps (shown with dashed rectangles), two of which are side gaps.}
        \label{fig:three graphs}
\end{figure}
\subparagraph*{Related work.}
The well-known Sugiyama framework \cite{sugiyama1981} for layered graph drawing serves as the main motivation of this work. As mentioned, a key step of this framework is to minimize crossings between two adjacent layers by permuting the order of nodes of one layer while keeping the second layer fixed, which is a known \NP-hard problem called one-sided crossing minimization (\oscm) \cite{eades1994oscmNPcomplete}. There exist heuristics with approximation guarantees \cite{eades1994oscmNPcomplete,sugiyama1981,DBLP:conf/gd/Nagamochi03}, FPT-algorithms parameterized by the number of crossings \cite{dujmovic2004oscmNPcompleteReference,DBLP:journals/ipl/KobayashiT16}, and exact algorithms based on integer linear programs \cite{DBLP:conf/gd/JungerM95}.

A restricted variant of \oscm\ has already been studied by Forster \cite{forster2004constrained}, where the relative order of node pairs can be restricted; thus the computed order has to conform to a given partial order. This is different to restrictions on gaps, which cannot be represented by partial orders. Further,  Nöllenburg and Wallinger \cite{DBLP:journals/jgaa/NollenburgW24} have considered gaps in a circular drawing style of graphs, called \emph{hive plots}. Our theoretical results are of independent interest to their work, and we extend their setting of gaps at fixed positions to gaps at arbitrary positions.

Gaps can also be regarded as groups of edges that can be bundled together. Edge bundling has already been applied in the context of layered graph drawing \cite{pupyrev2011edgeBundling}.

\subparagraph*{Structure.}
We state the formal problems for our \oscm-variants in \cref{sec:prelim}. In \cref{sec:oscmsg} and \cref{sec:oscmkg} we give polynomial time approximation algorithms for the respective problems. Exact ILP formulations are given in \cref{appendix:ilp}, and an experimental evaluation of selected algorithms is presented in \cref{section:experiments}. The source code is available online~\cite{Dobler_Roithinger_2025}.

\section{Preliminaries}\label{sec:prelim}
\subparagraph*{Permutations.} We treat permutations $\pi$ as lists of a set $X$. Two permutations $\pi,\pi'$ of disjoint sets can be concatenated by $\pi\star \pi'$. For $x,y\in \pi$ we write $x\prec_\pi y$ if $x$ comes before $y$ in $\pi$. For $X'\subseteq X$, $\pi[X']$ is the \emph{induced} permutation on $X'$, i.e., for all $x,y\in X'$, $x\prec_\pi y$ iff.\ $x\prec_{\pi[X']}y$.
Further, $\Pi(X)$ denotes the set of all permutations of $X$, and $\pi[i:j]$ is the set of elements in $\pi$ whose index is between $i$ and $j$ inclusively, using $1$-indexing.

\subparagraph*{One-sided crossing minimization.}
The problems discussed in this paper have as input a bipartite graph $G=(V_1\dot{\cup} V_2, E)$, $E\subset V_1\times V_2$. We set $n:=|V_1\cup V_2|$ and $m:=|E|$.
The classic one-sided crossing minimization (\oscm) problem is given $G$ and a permutation $\pi_1$ of $V_1$. The task is to find a permutation $\pi_2$ of $V_2$ that minimizes the number of edge crossings in a two-layer straight-line drawing of $G$ such that nodes in $V_1$ are ordered according to $\pi_1$ on the bottom layer and nodes in $V_2$ are ordered according to $\pi_2$ on the top layer. Such crossings can be determined combinatorially by $\pi_1$ and $\pi_2$; namely, edges $e=(u_1,u_2)$ and $e=(v_1,v_2)$, $u_i,v_i\in V_i$, cross w.r.t.\ $\pi_1$ and $\pi_2$ if $u_1\prec_{\pi_1}v_1$ and $u_2\succ_{\pi_2}v_2$, or $u_1\succ_{\pi_1}v_1$ and $u_2\prec_{\pi_2}v_2$.
Let $\crossings(G,\pi_1,\pi_2)$ be the number of crossings determined in such a way. Given $G$ and $\pi_1$, \oscm\ asks for $\pi_2$ minimizing $\crossings(G,\pi_1,\pi_2)$.
Throughout the paper we assume $\pi_1$ as fixed.
By slightly abusing the notation of $\crossings$, we furthermore define for $S,S'\in V_2$, $S\cap S'=\emptyset$ the value $\crossings(G,S,S')$ as follows. Let $\pi_2$ be any permutation such that all nodes in $S$ come before all nodes in $S'$. The value $\crossings(G,S,S')$ is the number of pairs $e,e'$ that cross w.r.t.\ $\pi_1$ and $\pi_2$ such that $e\cap S\ne \emptyset$ and $e'\cap S'\ne \emptyset$. %

We extend the \oscm\ problem by restricting the amount of allowed \emph{gaps}. For this, we note that $V_2$ consists of the disjoint union of $\real$ and $\dummy$, where $\real$ is the set of \emph{real} nodes and $\dummy$ is the set of \emph{dummy} nodes obtained by the preprocessing steps performed by the Sugiyama framework \cite{sugiyama1981}.
It is important to note that dummy nodes have degree one, which we exploit in all our algorithms.
A gap in $\pi_2$ is a maximal consecutive sequence of dummy nodes, and $\gaps(\pi_2)$ is the amount of gaps in $\pi_2$. Furthermore, a \emph{side gap} is a gap that either contains the leftmost or rightmost dummy nodes in $\pi_2$, $\pi_2$ is a \emph{side-gap permutation} if all of its gaps are side-gaps.
In our restricted \oscm\ variants, we either allow only side gaps in $\pi_2$, or at most $k$ gaps overall. The formal problems are given below, starting with side gaps.
\begin{problem}[\oscmsg]
    Given a bipartite graph $G=(V_1\dot{\cup}V_2,E)$ and a permutation $\pi_1$ of $V_1$, find a permutation $\pi_2\in \Pi(V_2)$ such that $\pi_2$ is a side-gap permutation and $\crossings(G_1,\pi_1,\pi_2)$ is minimal.
\end{problem}
\begin{problem}[\oscmkg]
    Given a bipartite graph $G=(V_1\dot{\cup}V_2,E)$, a permutation $\pi_1$ of $V_1$, and $k\in\mathbb{N}$, find a permutation $\pi_2\in \Pi(V_2)$ such that $\gaps(\pi_2)\le k$ and $\crossings(G_1,\pi_1,\pi_2)$ is minimal.
\end{problem}
Clearly, both problems are \NP-hard as they are equivalent to classic \oscm\ which is \NP-hard~\cite{eades1994oscmNPcomplete}, if we have no dummy nodes. %
\section{Approximation Algorithms for \oscmsg}\label{sec:oscmsg}
We show that any approximation algorithm for the classic \oscm\ problem can be transformed to an approximation algorithm for \oscmsg\ with the same approximation ratio. First, we show in the below lemma that there will never be edge crossings that involve two edges that are both incident to a dummy node in an optimal solution to \oscmsg\ and \oscmkg.
\begin{restatable}{lem}{lemswapdummy}
    \label{lemma:ordvirt}
    Given $\pi_1,\pi_2$ such that a pair of edges $e,e'\in V_1\times\dummy$ crosses, there is $\pi_2'$ such that (1) $\crossings(G,\pi_1,\pi_2')<\crossings(G,\pi_1,\pi_2)$, (2) $\gaps(\pi_2')\le \gaps(\pi_2)$, and (3) if $\pi_2$ is a side-gap permutation, so is $\pi_2'$.
\end{restatable}
\begin{proof}
    Let $\pi_1,\pi_2$ be such that a pair of edges $e,e'\in V_1\times \dummy$ crosses. Let $e=(u,v)$ and let $e'=(u',v')$, hence $v,v'\in \dummy$. We show that exchanging $v$ and $v'$ in $\pi_2$ does not increase the number of crossings. The statement then follows by induction, repeatedly exchanging such $v$ and $v'$ in $\pi_2$.
    Let $\pi_2'$ be the permutation after exchanging $v$ and $v'$.
    We consider different cases (note that $v$ and $v'$ have degree one).
    \begin{itemize}
        \item An edge $e''$ that crosses with both $e$ and $e'$ with respect to $\pi_1$ and $\pi_2'$, also crosses with both with respect to $\pi_1$ and $\pi_2$.
        \item An edge $e''$ that crosses with $e$ with respect to $\pi_1$ and $\pi_2'$, crosses with $e'$ with respect to $\pi_1$ and $\pi_2$.
        \item An edge $e''$ that crosses with $e'$ with respect to $\pi_1$ and $\pi_2'$, crosses with $e$ with respect to $\pi_1$ and $\pi_2$.
        \item The edges $e$ and $e'$ do not cross w.r.t. $\pi_1$ and $\pi_2'$.
        \item Edge crossings between edges that are both not incident to neither $v$ nor $v'$ stay the same.
    \end{itemize}
    Hence, we did not increase the number of crossings by exchanging $v$ and $v'$ in $\pi_2$.
\end{proof}
Due to the above lemma, we fix in the rest of the paper $\pi_2^\text{dm}$ as the order of $\dummy$ determined by sorting $\dummy$ ascending by their neighbor's position in $\pi_1$. Dummy nodes with the same neighbor can be ordered arbitrarily.
If for any solution $\pi_2$, $\pi_2[\dummy]\ne \pi_2^\text{dm}$, we can transform $\pi_2$ into $\pi_2'$ having properties (1)-(3) of \cref{lemma:ordvirt} and  with $\pi_2'[\dummy]=\pi_2^\text{dm}$.

Now, given an algorithm $\mathcal{A}$ for \oscm\ with approximation ratio $\alpha$, we get an approximation algorithm for \oscmsg\ with the same approximation ratio, as given below.
\begin{restatable}{theorem}{thmoscmsg}\label{thm:oscmsg}
    Let $\mathcal{A}$ be an algorithm for \oscm\ with approximation ratio $\alpha$ and runtime $\mathcal{O}(f(n,m))$. Then there exists an algorithm $\mathcal{B}$ for \oscmsg\ with approximation ratio $\alpha$ and runtime $\mathcal{O}(f(n,m)+m)$.
\end{restatable}
\begin{proof}
    Given $\mathcal{A}$, we define $\mathcal{B}$.
    First, we apply $\mathcal{A}$ to the \oscm\ instance~$(G[V_1\cup \real], \pi_1)$ and obtains a permutation of $\real$.
    This determines the order $\pi_2[\real]$ of real nodes. It remains to place the nodes in $\dummy$ to the left or right of $\pi_2[\real]$.
    For a single node $u\in \dummy$, placing $u$ left of $\pi_2[\real]$ results in less crossings than placing it to right if $\crossings(\{u\},\real)<\crossings(\real,\{u\})$. 
    Notice that it follows that then $\crossings(\{u'\},\real)<\crossings(\real,\{u'\})$ for all $u'\in \dummy$ with $u'\prec_{\pi_2^{\text{dm}}}u$.
    Thus it is enough to find the rightmost $u$ in $\pi_2^{\text{dm}}$ such that $\crossings(\{u\},\real)<\crossings(\real,\{u\})$. Let $\pi_2^{dm}=A\star (u)\star B$, where $\pi_2^{dm}$ is defined by the neighbors of dummy nodes in $\pi_1$ ($A$ and $B$ can be empty).    
    We place $A\star (u)$ before $\pi_2[\real]$ and $B$ after $\pi_2[\real]$. 
    The node $u$ can be found by binary search. This can be done in $\mathcal{O}(m)$ time if we can query $\crossings(\{u\},\real)$ and $\crossings(\real,\{u\})$ in constant time, which is possible by precomputing the sum of degrees of nodes in each suffix and prefix of~$\pi_1$.
    It is left to show that the approximation ratio $\alpha$ remains. Consider an optimal solution $\pi_2^*$ to the \oscmsg\ instance $(G,\pi_1)$.
    Notice that a dummy node $u$ is in the left side gap in $\pi_2^*$ if and only if it is in the left side gap in the permutation $\pi_2$ computed by $\mathcal{B}$. This is the case as $\crossings(\{u\},\real)$ is independent of the order of real nodes and $\mathcal{B}$ placed $u$ such that it minimizes the number of crossings involving the edge incident to $u$. Due to $\mathcal{A}$ being and $\alpha$-approximation algorithm, crossings between edges $e,e'\in V_1\times\real$ are at most $\alpha$ times more in $\pi_2$ when compared with $\pi_2^*$. The remaining crossings sum up to the same amount. Thus, $\mathcal{B}$ is an $\alpha$-approximation.
\end{proof}
The key idea for the algorithm $\mathcal{B}$ is that we can in polynomial time compute the optimal placement of gap nodes and this placement is independent of the ordering of real nodes.
If $\mathcal{A}$ is for example the median heuristic \cite{eades1994oscmNPcomplete}, then \cref{thm:oscmsg} gives us a polynomial time $3$-approximation algorithm for \oscmsg. For an exact algorithm, we can substitute for $\mathcal{A}$ any exact algorithm for \oscm~($\alpha=1$) such as ILP formulations \cite{DBLP:conf/gd/JungerM95}.

\section{Approximation Algorithms for \oscmkg}\label{sec:oscmkg}
Adapting heuristics for \oscmkg\ is not as straight-forward. This is because once we have determined $\pi_2[\real]$, we have to consider all possibilities of inserting dummy nodes without having more than $k$ gaps. Furthermore, now the optimal placement of dummy nodes is dependent on $\pi_2[\real]$. We will only be able to extend \oscm\ heuristics with the following property.
\begin{definition}
    Let $\mathcal{A}$ be an algorithm for \oscm, $(G,\pi_1)$ be any instance of \oscm\ with $G=(V_1\dot{\cup} V_2,E)$. Consider a set of new nodes $V'$, $E'\subseteq V_1\times V'$, and $G'=(V_1\dot{\cup} (V_2\cup V'), E\cup E')$. Now apply $\mathcal{A}$ to $(G,\pi_1)$ and to $(G',\pi_1)$ to obtain solutions $\pi_2$ and $\pi_2'$, respectively. The algorithm $\mathcal{A}$ is \emph{dummy-independent} if $\pi_2[V_2]=\pi_2'[V_2]$ always holds.    
\end{definition}
Examples of dummy-independent algorithms are for example the barycenter-, and median-heuristic. By plugging $V'=\dummy$ in the above definition, we see that the order of real nodes computed by $\mathcal{A}$ is independent of the dummy nodes in $G$, when $\mathcal{A}$ is dummy-independent.

We can now extend any dummy-independent approximation algorithm $\mathcal{A}$ to \oscmkg\ maintaining the approximation ratio.
\begin{restatable}{theorem}{thmoscmkg}\label{thm:oscmkg}
    Let $\mathcal{A}$ be a dummy-independent algorithm for \oscm\ with approximation ratio $\alpha$ and runtime $\mathcal{O}(f(n,m))$. Then there exists an algorithm $\mathcal{B}$ for \oscmkg\ with approximation ratio $\alpha$ and runtime $\mathcal{O}(f(n,m)+|\real|\cdot|\dummy|^2\cdot k)$.
\end{restatable}
\begin{proof}
    The algorithm $\mathcal{B}$ first determines $\pi_2^\text{r}:=\pi_2[\real]$ by applying $\mathcal{A}$ to the \oscm\ instance~$(G[V_1\cup \real], \pi_1)$. We define a dynamic program to merge the two orders $\pi_2^{\text{r}}$ and $\pi_2^{\text{dm}}$.
    The dynamic programming table $DP$ contains entries $DP[g,i,j]$ which represents the minimum number of crossings between edge pairs $e,e'$, $e$ being incident to $\real$ and $e'$ being incident to $\dummy$, using at most $g$ gaps when merging the first $i$ nodes in $\pi_2^\text{r}$ and the first $j$ nodes in $\pi_2^\text{dm}$; further, $g$ goes from $0$ to $k$, $i$ goes from $0$ to $|\real|$, and $j$ goes from $0$ to $|\dummy|$.
    The base cases are 
    \begin{itemize}
        \item $DP[0,i,0]=0$ for $0\le i\le |\real|$,
        \item $DP[0,i,j]=\infty$ for $0\le i\le |\real|$, $1\le j\le |\dummy|$,
    \end{itemize}
    and the transitions for $g>0$ can be computed as
    \begin{align*}
        DP[g,i,j] =& \min_{0\le j'\le j}[DP[g-1,i,j']\\
        &+\crossings(G,\pi_2^\text{r}[1:i],\pi_2^\text{dm}[j'+1:j])+\crossings(G,\pi_2^\text{dm}[j'+1:j], \pi_2^\text{r}[i+1:|\real|])].
    \end{align*}
    The optimal number of crossings can be read from $DP[k,|\real|,|\dummy|]$, and the corresponding permutation can be reconstructed from the entries in $DP$.
    The runtime can be achieved by precomputing $\crossings(G,\pi_2^\text{r}[1:i],\pi_2^\text{dm}[j'+1:j])$ and $\crossings(G,\pi_2^\text{dm}[j'+1:j], \pi_2^\text{r}[i+1:|\real|])$.

    For correctness, consider an optimal solution $\pi_2^{\text{opt}}$ with $c$ crossings. By \cref{lemma:ordvirt} no edge pairs incident to $\dummy$ cross. Now contract each set of dummy nodes that appear in a gap together, obtaining the graph $G'$. Apply $\mathcal{A}$ to $(G',\pi_1)$, obtaining a solution $\pi'_2$ with at most $\alpha\cdot c$ crossings. Now revert the contraction and replace each contracted node by its original sequence of dummy nodes in $\pi'_2$. The newly obtained permutation is in the solution space of the dynamic program because $\pi_2^{\text{r}}=\pi_2'[\real]$ ($\mathcal{A}$ is dummy-independent), hence we have an $\alpha$-approximation.
\end{proof}

\section{Integer Linear Programs}\label{appendix:ilp}
In order to compare the results of the heuristic algorithms to an optimal solution, we defined an integer linear program (ILP) for the \oscmkg\ problem.
There is no need to develop a specialized formulation for \oscmsg\, since we can determine the relative order of real nodes using an existing ILP formulation, for example the one defined in \cite{DBLP:conf/gd/JungerM95}, and apply the algorithm described in \cref{thm:oscmsg}.

\subsection{\oscmkg}\label{sec:ilpkg}

Exact algorithms for \oscm\ are not dummy-independent. Hence, \cref{thm:oscmkg} cannot be applied to such algorithms. Hence, we propose an exact ILP formulation for \oscmkg.

The key insight which will make this formulation efficient is \cref{lemma:ordvirt}, which already determines the order $\pi_2^\dummySuperscript$ of dummy nodes. We will count the number of times there are real nodes between a pair of adjacent (in $\pi_2^\dummySuperscript$) dummy nodes, which will correspond to  the number of gaps minus one.
Thus let $N=\{(u,v)\mid u,v\in \dummy, u\text{ is the predecessor of }v\text{ in }\pi_2^\dummySuperscript\}$.
For an instance $(G,\pi_1,k)$ of \oscmkg, let $V_2=\{v_1,v_2,\dots,v_p\}$.
The formulation (given in \cref{ilp:oscmkg}) makes use of the following variables with the given semantics.
\begin{itemize}
    \item A binary variable $x_{ij}$ for each $1\le i,j\le p, i\ne j$. We have $x_{ij} = 1$ if and only if $\permComesBefore{\pi_2}{v_i}{v_j}$.
    \item A binary variable $g_{ij}$ for each $(v_i,v_j)\in N$. We have $g_{ij}=1$ if there is a real node between $v_i$ and $v_j$ in $\pi_2$.
\end{itemize}

In all equations below, assume that $i,j,k$ are pairwise different.
\begin{ilp}[\oscmkg]\label{ilp:oscmkg} \ \\ 
    \begin{align}
    \text{minimize:}\quad  \sum_{v_i, v_j \in V_2} x_{ij} \cdot \crossings(G,\{v_i\},\{v_j\}) \\
    \text{subject to:}\quad x_{ij} + x_{ji}          = 1 &  \ilpEquationSpacing v_i, v_j \in V_2 \label{ilp:k-gaps:constr:bool-order-mutex} \\
        x_{ij} + x_{jk} + x_{ik} \le 3 & \ilpEquationSpacing v_i, v_j, v_k  \in V_2 \label{ilp:k-gaps:constr:order-transitivity} \\
       x_{ij}=1&\ilpEquationSpacing (v_i,v_j)\in N\label{ilp:k-gaps:constr:ordervirtual}\\
       g_{ij}\ge x_{ik}+x_{kj}-1&\ilpEquationSpacing (v_i,v_j)\in N, v_k\in \real\label{ilp:k-gaps:constr:gapconstr}\\
       \sum_{(v_i,v_j)\in N}g_{ij}\le k-1&\label{ilp:k-gaps:constr:gapcount}
    \end{align}
\end{ilp}

The constraints in \cref{ilp:oscmkg} have the following meanings:
\begin{itemize}
    \item \eqref{ilp:k-gaps:constr:bool-order-mutex} ensures the anti-symmetry of $\pi_2$.
    \item \eqref{ilp:k-gaps:constr:order-transitivity} ensures the transitivity of $\pi_2$.
    \item \eqref{ilp:k-gaps:constr:ordervirtual}  ensures that the order of virtual nodes is as in \cref{lemma:ordvirt}.
    \item \eqref{ilp:k-gaps:constr:gapconstr} ensures that $g_{ij}$ assumes 1 if there is a real node $v_k$ between $v_i$ and $v_j$ in $\pi_2$.
    \item \eqref{ilp:k-gaps:constr:gapcount} limits the total amount of gaps to $k$.
\end{itemize}

The objective function minimizes the total number of edge crossings by summing the values $\crossings(G,\{v_i\},\{v_j\})$ for $i,j$ where $x_{ij} = 1$, where $\crossings(G,\{v_i\},\{v_j\})$ is the number of edge crossings between edge pairs incident to $v_i$ and $v_j$ when $v_i$ is placed before $v_j$.
The permutation $\pi_2$ can be read off from the values of the $x_{ij}$ variables. That is, $v_i\prec_{\pi_2}v_2$ iff.\ $x_{ij}=1$.

\section{Experiments}\label{section:experiments}

To test the practical application of the algorithms presented in this paper, we have decided to implement them using Python.

\subparagraph*{Test data.}
To test the performance of the algorithms in terms of execution time and in terms of edge crossing count after minimization, ``random'' bipartite graphs of various dimensions were generated.
We have outlined the graph generation process in \cref{alg:random-graph-gen}.

\newcommand{\varTotalNodesPerLayer}{\ensuremath{n}}
\newcommand{\varDummyNodeRatio}{\ensuremath{f_{\dummySuperscript}}}
\newcommand{\varAverageNodeDegree}{\ensuremath{\mathit{deg}_{\text{avg}}}}
\newcommand{\varDummyNodesPerLayer}{\ensuremath{n_{\dummySuperscript}}}
\newcommand{\varRealNodesPerLayer}{\ensuremath{n_{\realSuperscript}}}

\begin{algorithm}[t]
    \KwIn{A positive integer $\varTotalNodesPerLayer$ specifying how many nodes should be generated for each layer. A floating point number $0 \le \varDummyNodeRatio \le 1$ specifying what fraction of nodes generated should be dummy. A positive number $\varAverageNodeDegree$ specifying the average degree a real node should have.}
    \KwOut{A ``random'' two layered graph $G$ with the given specifications.}
    $\varDummyNodesPerLayer = \floor{\varTotalNodesPerLayer \cdot \varDummyNodeRatio }$ \;
    $\varRealNodesPerLayer = \varTotalNodesPerLayer - \varDummyNodesPerLayer$ \;
    initialize an empty 2-layered graph $G$ \;
    add $\varDummyNodesPerLayer$ real nodes to both layers of $G$ \;
    choose $\floor{\varRealNodesPerLayer \cdot max(\varRealNodesPerLayer, \varAverageNodeDegree)}$ edges from $V_1^{\realSuperscript} \times \vTwoReal $ uniformly at random and add them to $G$ \;
    create $\varDummyNodesPerLayer$ dummy nodes in the first and second layer \;
    add edges from each of these dummy nodes to a random other node in the opposite layer, ensuring that dummy nodes only have a single incident edge \;
    \Return{$G$} \;
    \caption{GENERATE-RANDOM-2-LAYER-GRAPH}
    \label{alg:random-graph-gen}
\end{algorithm}

\subparagraph*{Algorithms.}
We tested a set of $6$ algorithms for one-sided crossing minimization with gap constraints. These include the following:
\begin{description}
    \item[\textit{median sidegaps}] Algorithm of \cref{thm:oscmsg} applied to the median heuristic \cite{eades1994oscmNPcomplete}. 
    \item[\textit{barycenter sidegaps}] Algorithm of \cref{thm:oscmsg} applied to the barycenter heuristic \cite{sugiyama1981}.
    \item[\textit{ilp sidegaps}] Algorithm of \cref{thm:oscmsg} applied to an exact ILP formulation from \cite{DBLP:conf/gd/JungerM95}.
    \item[\textit{median kgaps}] Algorithm of \cref{thm:oscmkg} applied to the median heuristic \cite{eades1994oscmNPcomplete}. 
    \item[\textit{barycenter kgaps}] Algorithm of \cref{thm:oscmkg} applied to the barycenter heuristic \cite{sugiyama1981}.
    \item[\textit{ilp kgaps}] The ILP formulation from \cref{sec:ilpkg} for \oscmkg.
\end{description}
All ILP formulations were implemented in Gurobi, while transitivity constraints \cref{ilp:k-gaps:constr:order-transitivity} were implemented as ``lazy constraints '' that will be introduced by Gurobi using branch and cut.

\subparagraph*{Setup.}
Each experiment was executed with 20 random graph instances, with parameters set to the following values (unless varied): 40 nodes per layer, a dummy node fraction of 0.2, an average node degree of 3, using the median heuristic.
We set a time out of 5 minutes, however, no timeouts were recorded.

The tests were executed on a compute cluster with ``Intel Xeon E5-2640 v4'' processors running Ubuntu 18.04.6.
The Python version used is 3.12.4.
The integer linear programs were written using the Gurobi python API and solved by a local instance of the v11.0.1 Gurobi solver with the thread limit set to 1.

\subparagraph*{Results.}
For the following plots the $x$-axis is labeled with the varied parameter, while the other parameters remain fixed.
For experiments that include the ILP algorithm, there are four plots, one plotting the varied parameter against runtime, one plotting the varied parameter against the total crossing counts, and the same again but in terms of ratio compared to the optimal ILP solution. That is, the $2$-gap heuristics are compared to the $2$-gap ILP and the side-gap heuristics are compared to the side-gap ILP.
Plots with ``time'' on the $y$-axis which include ILP runs have a logarithmic $y$-axis.

\newcommand{\currExperimentVersion}{temp}
\newcommand{\experimentNameTypeToFileName}[2]{jakobthesis/graphics/experiment_plots/testcase_\currExperimentVersion_#1_#2.pdf}
\newcommand{\experimentSubFigure}[3]{
	\begin{subfigure}[b]{0.5\columnwidth}
		\centering
		\includegraphics[width=\textwidth]{\experimentNameTypeToFileName{#1}{#2}}
		\caption{#3}
		\label{fig:experiment:#1-#2}
	\end{subfigure}
}
\newcommand{\experimentFullPlots}[2]{

	\begin{figure}           
		\experimentSubFigure{#1}{crossings}{Crossing count.}
		\hfill %
		\experimentSubFigure{#1}{ratio_crossings}{Crossing ratio.}

		\experimentSubFigure{#1}{time_s}{Time.}
		\hfill %
		\experimentSubFigure{#1}{ratio_time_s}{Time ratio.}
		
		\caption{#2}
		\label{fig:experiment:#1}
	\end{figure}
}

\newcommand{\experimentNoRatioPlots}[2]{

	\begin{figure}
		\experimentSubFigure{#1}{crossings}{Crossings.}
		\hfill %
		\experimentSubFigure{#1}{time_s}{Time.}
		\caption{#2}
		\label{fig:experiment:#1}
	\end{figure}
}

\experimentFullPlots{2_gaps_vs_side_gaps}{Plots comparing the algorithms for \oscmsg\ with algorithms for \oscmkg\ with $k=2$.}
\cref{fig:experiment:2_gaps_vs_side_gaps} compares results for the side-gaps approach and the $k$-gaps approach with $k=2$.
Both ILP formulations never time out and show good performance with regard to runtime, at least for our variant of generating random instances.
It would be reasonable to expect that 2 arbitrary gaps would result in significantly fewer crossings compared to using only side-gaps; however, the difference is quite small.

\experimentFullPlots{k_gaps_count_variation}{Plots for \oscmkg\ with varying gap count $k$.}
\cref{fig:experiment:k_gaps_count_variation} compares results for different gap counts for the \oscmkg\ problem.
As can be seen in \cref{fig:experiment:k_gaps_count_variation-crossings}, increasing gap count yields diminishing returns.
That is, the crossings are drastically reduced when comparing two gaps with one gap. Increasing the gap count more does not increase the number of crossings, even when allowing arbitrary many gaps.

\section{Conclusion}
Further research is required to properly integrate our algorithms into the Sugiyama framework. In particular, adjustments might be required to guarantee few edge crossings over all layers, not just between a pair of layers. One might also investigate larger instances, i.e., with more vertices or higher edge densities.
Additionally, case studies could show how few gaps can reduce the amount of clutter in the layered graph drawing.

\bibliography{literature}

\end{document}